\documentclass[DIV9,a4paper,final]{scrartcl}

\usepackage[utf8]{inputenc}
\usepackage[T1]{fontenc}

\usepackage{amssymb}
\usepackage{amsthm}
\usepackage{booktabs}
\usepackage{enumitem}
\usepackage{float}
\usepackage{ifthen}
\usepackage{lmodern}
\usepackage{microtype}
\usepackage{multicol}
\usepackage{nccmath}
\usepackage{tikz}

\usetikzlibrary{calc,chains,decorations.pathmorphing}

\newtheorem{theorem}{Theorem}
\newtheorem{lemma}[theorem]{Lemma}
\newtheorem{corollary}[theorem]{Corollary}
\newtheorem{proposition}[theorem]{Proposition}

\setlist{itemsep=1pt,parsep=0pt,topsep=2pt}

\renewcommand{\geq}{\geqslant}
\renewcommand{\leq}{\leqslant}
\renewcommand{\ge}{\geq}
\renewcommand{\le}{\leq}

\newcommand{\autA}{\mathcal{A}}
\newcommand{\bigO}[1]{\ensuremath{\mathcal{O}({#1})}}
\newcommand{\PSPACE}{\textsf{PSPACE}}

\newcommand{\union}{\mathbin{\cup}}
\newcommand{\intersect}{\mathbin{\cap}}

\newcommand{\abs}[1]{\ensuremath\left|#1\right|}
\newcommand{\floor}[1]{\ensuremath\left\lfloor#1\right\rfloor}
\newcommand{\ceil}[1]{\ensuremath\left\lceil#1\right\rceil}
\newcommand{\set}[2]{\ensuremath{\left\{#1 \mid #2\right\}}}
\newcommand{\os}[1]{\ensuremath{\left\{#1\right\}}}
\newcommand{\gR}{\ensuremath{\mathcal{R}}}
\newcommand{\gL}{\ensuremath{\mathcal{L}}}
\newcommand{\gJ}{\ensuremath{\mathcal{J}}}
\newcommand{\gK}{\ensuremath{\mathcal{K}}}
\newcommand{\Req}{\mathrel{\gR}}

\newcommand{\Keq}{\mathrel{\gK}}
\newcommand{\Rle}{\mathrel{\le_\gR}}
\newcommand{\Lle}{\mathrel{\le_\gL}}
\newcommand{\Jle}{\mathrel{\le_\gJ}}
\newcommand{\Rl}{\mathrel{<_\gR}}
\newcommand{\Ll}{\mathrel{<_\gL}}
\newcommand{\Jl}{\mathrel{<_\gJ}}
\newcommand{\Rg}{\mathrel{>_\gR}}

\newcommand{\icz}{\iota_{\mathsf{=0}}}
\newcommand{\ico}{\iota_{\mathsf{=1}}}
\newcommand{\idec}{\iota_{\mathsf{dec}}}
\newcommand{\iinc}{\iota_{\mathsf{inc}}}
\newcommand{\irotl}{\iota_{\mathsf{rotl}}}
\newcommand{\irotr}{\iota_{\mathsf{rotr}}}
\newcommand{\imvl}{\iota_{\mathsf{mvl}}}
\newcommand{\imvr}{\iota_{\mathsf{mvr}}}
\newcommand{\isync}{\iota_{\mathsf{sync}}}
\newcommand{\idisp}{\iota_{\mathsf{off}}}
\newcommand{\ival}[1]{\iota_{\mathsf{val={#1}}}}

\newcommand{\Lco}{L_{\mathsf{=1}}}
\newcommand{\Lrotl}{L_{\mathsf{rotl}}}
\newcommand{\Lrotr}{L_{\mathsf{rotr}}}
\newcommand{\Linc}{L_{\mathsf{inc}}}
\newcommand{\Ldec}{L_{\mathsf{dec}}}
\newcommand{\Lreset}{L_{\mathsf{reset}}}

\newcommand{\ov}[1]{\overline{#1}}

\newcommand{\ie}{i.e.,~}
\newcommand{\eg}{e.g.~}

\let\oldpar\paragraph
\renewcommand{\paragraph}[1]{\oldpar*{\bf #1}}

\makeatletter
\DeclareOldFontCommand{\it}{\normalfont\itshape}{\mathit}
\makeatother

\title{Green's Relations in \\ Finite Transformation Semigroups}
\author{Lukas Fleischer \and Manfred Kuf\-leitner}
\date{FMI, University of Stuttgart\thanks{This work was supported by the DFG grants DI 435/5-2 and \mbox{KU 2716/1-1}.}\\
  Universitätsstraße 38, 70569 Stuttgart, Germany\\
  \normalsize\texttt{\{fleischer,kufleitner\}@fmi.uni-stuttgart.de}}

\begin{document}

\maketitle

\begin{abstract}
  \noindent
  {\sffamily\normalsize\bfseries{Abstract.}} \ 
  We consider the complexity of Green's relations when the semigroup is given
  by transformations on a finite set.
  Green's relations can be defined by reachability in the
  (right/left/two-sided) Cayley graph. The equivalence classes then correspond
  to the strongly connected components. It is not difficult to show that, in
  the worst case, the number of equivalence classes is in the same order of
  magnitude as the number of elements. Another important parameter is the
  maximal length of a chain of components. Our main contribution is an
  exponential lower bound for this parameter. There is a simple construction
  for an arbitrary set of generators. However, the proof for constant alphabet
  is rather involved. Our results also apply to automata and their syntactic
  semigroups.
 \end{abstract}

\section{Introduction}

Let $Q$ be a finite set with $n$ elements. There are $n^n$ mappings from $Q$ to
$Q$. Such mappings are called \emph{transformations} and the elements of $Q$
are called \emph{states}. The composition of transformations defines an
associative operation. If $\Sigma$ is some arbitrary subset of transformations,
we can consider the \emph{transformation semigroup} $S$ generated by $\Sigma$;
this is the closure of $\Sigma$ under composition.\footnote{When introducing
transformation semigroups in terms of actions, then this is the framework of
\emph{faithful} actions.} The set of all transformations on $Q$ is called the
\emph{full transformation semigroup} on $Q$. One can view $(Q,\Sigma)$ as a
description of $S$. Since every element $s$ of a semigroup~$S$ defines a
transformation $x \mapsto x \cdot s$ on $S^1 = S \cup \os{1}$, every
semigroup~$S$ admits such a description $(S^1,S)$; here, $1$ either denotes the
neutral element of $S$ or, if $S$ does not have a neutral element, we add $1$
as a new neutral element. Essentially, the description $(S^1,S)$ is nothing but
the multiplication table for~$S$. On the other hand, there are cases where a
description as a transformation semigroup is much more succinct than the
multiplication table. For instance, the full transformation semigroup on $Q$
can be generated by a set $\Sigma$ with three
elements~\cite{HolzerKoenig2004tcs}. In addition to the size of $S$, it would
be interesting to know which other properties could be derived from the number
$n$ of states.

Green's relations are an important tool for analyzing the structure of a
semigroup $S$. They are defined as follows:
\begin{equation*}
  s \Rle t \text{\,~if~\,} s S^1 \subseteq t S^1, \ \ 
  s \Lle t \text{\,~if~\,} S^1 s \subseteq S^1 t, \ \ 
  s \Jle t \text{\,~if~\,} S^1 s S^1 \subseteq S^1 t S^1.
\end{equation*}
We write $s \Req t$ if both $s \Rle t$ and $s \Rle t$; and we set $s \Rl t$ if
$s \Rle t$ but not $s \Req t$. The relations $\gL$, $\Ll$, $\gJ$ and $\Jl$ are
defined analogously. The relations $\gR$, $\gL$, and $\gJ$ form equivalence
relations. The equivalence classes corresponding to these relations are called
\emph{$\gR$-classes} (resp.\ \emph{$\gL$-classes}, \emph{$\gJ$-classes}) of
$S$. Instead of ideals, one could alternatively also use reachability in the
right (resp.\ left, two-sided) Cayley graph of $S$ for defining $\Rle$ (resp.\
$\Lle$, $\Jle$). We note that $s \Rl t$ implies $s \Jl t$ and, symmetrically,
$s \Ll t$ implies $s \Jl t$.
The complexity of deciding Green's relations for transformation semigroups was
recently shown to be \PSPACE-complete~\cite{BrandlSimon15dlt}.
When considering a transformation semigroup on $n$ states, one of our first
results shows that the maximal number of $\gJ$-classes is in $n^{\Theta(n)}$.
In particular, the number of equivalence classes is in the same order of
magnitude as the size of the transformation semigroup. Since every $\gJ$-class
contains at least one $\gR$- and one $\gL$-class, the same bound holds for
$\gR$ and~$\gL$.

Another important parameter is the maximal length $\ell$ such that there are
elements $s_1,\ldots,s_\ell$ with $s_1 \Rg \cdots \Rg s_\ell$, called the
\emph{$\gR$-height}. Similarly, we are interested in the $\gL$- and
$\gJ$-height. Many semigroup constructions such as the Rhodes expansion and
variants thereof rely on this parameter; see
e.g.~\cite{CartonMichel03tcs,eil76:short,GanardiEtAl2016fsttcs}. We show that
the maximal $\gR$-height is in $2^{\Theta(n)}$; for the maximal $\gL$-height
and $\gJ$-height we only have $2^{\Omega(n)}$ as a lower bound. Proving the
lower bounds for a fixed number of generators is much more involved than for
arbitrarily many generators.
The exponential lower bounds are quite unexpected in the following sense: If
the transformation semigroup is small, then the number of equivalence classes
(and hence, the lengths of chains) cannot be big. On the other hand, the
transformation semigroup is maximal if it is full. And an equivalence class in
the full transformation semigroup only depends on the number of states in the
image; this is because we can apply arbitrary permutations. In particular, the
number of equivalence classes in these two extreme cases is small.

There is a tight connection between deterministic automata and transformation
semigroups. Roughly speaking, a transformation semigroup is an automaton
without initial and finial states. The main difference is that for automata,
one usually is interested in the syntactic semigroup rather than the
transformation semigroup; the syntactic semigroup is the transformation
semigroup of the minimal automaton.
We show that the above bounds on the number of equivalence classes and heights
also apply to syntactic semigroups. 

\begin{theorem}\label{thm:numclasses}
  For each $n \in \mathbb{N}$, there exists a minimal automaton $\autA_n$ with
  $n$ states over an alphabet of size $5$ such that the number of $\gJ$-classes
  (resp.\ $\gR$-classes, $\gL$-classes) of the transformation semigroup
  $T(\autA_n)$ is at least $(n-4)^{n-4}$.
\end{theorem}

\begin{theorem}\label{thm:main}
  There exists a sequence of minimal automata $(\autA_n)_{n \in \mathbb{N}}$
  over a fixed alphabet such that $\autA_n$ has $n$ states and the $\gR$-height
  (resp.\ $\gL$-height, $\gJ$-height) of the transformation semigroup
  $T(\autA_n)$ is in $\Omega(2^n / n^{9.5})$.
\end{theorem}

\section{Preliminaries}

A \emph{semigroup} is a set $S$ equipped with an associative operation
${\cdot} \colon S \times S \to S$. A \emph{subsemigroup} of $S$ is a subset $T$
such that $s_1 s_2 \in T$ for all $s_1, s_2 \in T$. It is called
\emph{completely isolated} if the converse implication holds, \ie{}$s_1 s_2 \in
T$ implies $s_1 \in T$ and $s_2 \in T$ for all $s_1, s_2 \in S$. 
The \emph{opposite} semigroup of $S$ is obtained by replacing the operation
with its left-right dual ${\circ} \colon S \times S \to S$ defined by $x \circ
y = y \cdot x$.

In general, Green's relations in a subsemigroup $T$ of $S$ do not coincide with
the corresponding relations in $S$. However, if $T$ is a completely isolated
subsemigroup, the following property holds:
\begin{proposition}\label{prop:isolated}
  Let $S$ be a semigroup and let $T$ be a completely isolated subsemigroup of
  $S$. Let $\gK$ be one of the relations $\le_\gR$, $\le_\gL$, $\le_\gJ$,
  $\gR$, $\gL$ or $\gJ$. Then, for all $x, y \in T$, we have $x \Keq y$ in $S$
  if and only if $x \Keq y$ in $T$.
\end{proposition}
\begin{proof}
  We will only prove the statement for the preorder $\Rle$.
  For the implication from right to left, we have $xS^1 \subseteq xT^1 S^1
  \subseteq yT^1 S^1 \subseteq y S^1$.
  For the converse implication, suppose that $xS^1 \subseteq yS^1$, \ie{}there
  exists some $z \in S^1$ such that $yz = x$. Since $T$ is completely isolated,
  we have $z \in T^1$, which yields $zT^1 \subseteq T^1$ and thus, $xT^1 =
  yzT^1 \subseteq yT^1$.
\end{proof}
An \emph{$\gR$-chain} is a sequence $(s_1, \dots,
s_\ell)$ of elements of $S$ such that $s_{i+1} \Rl s_i$ for all $i \in \os{1,
\dots, \ell-1}$; $\ell$ is called the \emph{length} of the $\gR$-chain. The
maximal length of an $\gR$-chain of $S$ is called the \emph{$\gR$-height} of $S$.
The notions $\gL$-chains, $\gJ$-chains, $\gL$-height and $\gJ$-height are
defined analogously.

A \emph{partial transformation} on a set $Q$ is a partial function $f \colon Q
\to Q$. If the domain of $f$ is all of $Q$, \ie{}if $f$ is a total function,
$f$ is called a \emph{transformation}.
A partial transformation $f \colon Q \to Q$ is called \emph{injective} if $f(p)
\ne f(q)$ whenever $p\neq q$ and both $f(p)$ and $f(q)$ are defined. The
elements of $Q$ are often called \emph{states}.
In the following, we use the notation $q \cdot f$ instead of $f(q)$ to denote
the image of an element $q \in Q$ under $f$. For $R \subseteq Q$ let $R \cdot f
= \set{q \cdot f}{q \in R}$. Note that for all subsets $R \subseteq Q$ and all
partial transformations $f \colon Q \to Q$, the inequality $\abs{R \cdot f} \le
\abs{R}$ holds; we will implicitly use this property throughout the paper.
The \emph{composition} $fg$ of two transformations $f \colon Q \to Q$ and $g
\colon Q \to Q$ is defined by $q \cdot fg = (q \cdot
f) \cdot g$. The composition is associative.

The set of all partial transformations (resp.\ transformations) on a fixed set
$Q$ forms a semigroup with composition as the binary operation. It is called
the \emph{full partial transformation semigroup} (resp.\ \emph{full
transformation semigroup}) on~$Q$. Subsemigroups of full (partial)
transformation semigroups are called \emph{(partial) transformation semigroups}
and are often specified in terms of generators. Partial transformation
semigroups and transformation semigroups are strongly related.
On one side, every transformation semigroup also is a partial transformation
semigroup. In the other direction a slightly weaker statement holds:
\begin{proposition}\label{prop:partial}
  Let $P$ be a partial transformation semigroup on $n$ states. Then there
  exists a transformation semigroup on $n+1$ states which is isomorphic to~$P$.
\end{proposition}
\begin{proof}
  Let $P$ be a partial transformation semigroup on a finite set $Q$, generated
  by a set $\Sigma$. Let $p$ be a new element not in $Q$. We extend the
  elements of $\Sigma$ to transformations on $Q \union \os{p}$ by letting $q
  \cdot a = p$ whenever $q \cdot a$ is undefined in $P$. In particular, we let
  $p \cdot a = p$ for all $a \in \Sigma$. By construction, the transformation
  semigroup generated by this extended set of generators is isomorphic to $P$.
\end{proof}
A partial transformation semigroup is called \emph{injective} if it is
generated by a set of injective partial transformations. An important property
of injective partial transformation semigroups is that they have a left-right
dual:
\begin{proposition}\label{prop:opposite}
  The opposite semigroup of an injective partial transformation semigroup is a
  partial transformation semigroup.
\end{proposition}
\begin{proof}
  Let $P$ be a partial transformation semigroup on a set of states $Q$,
  generated by a set of injective transformations $\Sigma$. Since the
  composition of injective partial transformations is injective, every element
  of $P$ is injective.
  For each $f \in P$ we let $\ov f \colon Q \to Q$ be the partial
  transformation that is undefined on $Q \setminus (Q \cdot f)$ and defined by
  $(q \cdot f) \cdot \ov f = q$ for all $q \in Q$. It is well-defined because
  $f$ is injective. Moreover, it is easy to check that $\ov{fg} = \ov g \ov f$
  for all $f, g \in P$.

  Let $\ov P$ be the partial transformation semigroup generated by $\ov \Sigma
  = \set{\ov a}{a \in \Sigma}$. Then, for all $f, g, h \in P$, we have $fg = h$
  in $P$ if and only if $\ov g \ov f = \ov h$ in $\ov P$, which shows that $\ov
  P$ is isomorphic to the opposite semigroup of $P$.
\end{proof}
Transformation semigroups naturally arise when considering deterministic finite
automata. Let $\autA = (Q, \Sigma, \delta, q_0, F)$ be a deterministic finite
automaton. Then, each letter $a \in \Sigma$ can be interpreted as a
transformation $a \colon Q \to Q$ where $q \cdot a = \delta(q, a)$. The
transformation semigroup on $Q$ generated by all letters in $\Sigma$ is denoted
by $T(\autA)$ and it is called the \emph{transition semigroup} of $\autA$.
Conversely, given a transformation semigroup $T$ on a finite set $Q$ and a
finite set of generators $\Sigma$, for each $q_0 \in Q$ and $F \subseteq Q$,
one can define a deterministic finite automaton $\autA = (Q, \Sigma, \delta,
q_0, F)$ where $\delta \colon Q \times \Sigma \to Q$ is defined as $\delta(q,
a) = q \cdot a$.

A well-known approach for translating bounds on the size of a transformation
semigroup to syntactic monoids is to \emph{make} an automaton minimal. This can
be done by introducing a new generator $c$ with $q_i \cdot c = q_{i+1}$ for $Q
= \os{q_1,\ldots,q_n}$ and $q_{n+1} = q_1$; moreover, one chooses some
arbitrary state to be both initial and final. We adapt this construction to
also work with Green's relations.

\begin{proposition}\label{prop:minimal}
  Let $T$ be a transformation semigroup on $n$ states, generated by a finite
  set $\Sigma$. Then there exists a minimal $(n+1)$-state deterministic finite
  automaton $\autA$ over an alphabet of size $\abs{\Sigma}+1$ such that $T$ is
  a completely isolated subsemigroup of $T(\autA)$.
\end{proposition}
\begin{proof}
  Let $T$ be a transformation semigroup on a set of states $Q =
  \os{q_1,\ldots,q_n}$, generated by $\Sigma$. Let $\autA = (Q \union \os{q_0},
  \Sigma \union \os{c}, \delta, q_0, \os{q_n})$ be the automaton defined by
  $\delta(q_0,a) = q_0$ and $\delta(q_i, a) = q_i \cdot a$ for $i \geq 1$ and
  all $a \in \Sigma$. The transitions for the letter $c$ are defined by
  $\delta(q_i,c) = q_{i+1}$ for $i<n$ and $\delta(q_n,c) = q_1$.
  This automaton is minimal: for two different states $q_i, q_j \in Q \cup
  \os{q_0}$ with $i>j$, we have $\delta(q_i, c^{n-i}) = q_n$ but $\delta(q_j,
  c^{n-i}) \ne q_n$. 

  By construction, $T$ is a subsemigroup of $T(\autA)$.
  To see that $T$ is completely isolated within $T(\autA)$, note that we have
  $\delta(q_0, u) = q_0$ if and only if $u \in \Sigma^*$.
\end{proof}

\section{Bounds for the Number of Classes}

Let $\gK$ be any of the relations $\gR$, $\gL$ or $\gJ$.
The na\"{\i}ve upper bound for the number of $\gK$-classes of a transformation
semigroup $T$ on $n$ states is given by the size of $T$ itself. Since there
are $n^n$ different functions from $Q$ to $Q$, the semigroup $T$ contains at
most $n^n$ elements. It is well known that this bound is tight even for a
constant number of generators, since for each $n \ge 1$ there exists a
transformation semigroup of size $n^n$ generated by a set $\Sigma$ with three
elements; see \eg\cite{HolzerKoenig2004tcs}.

As each $\gR$-class (resp. $\gL$-class, $\gJ$-class) consists of at least one
element, the number of such classes is also bounded by $n^n$. We now show that
this upper bound is tight up to a constant factor.
\begin{proposition}\label{prop:k-classes-lower}
  Let $T$ be a transformation semigroup on $n$ states, generated by a finite
  set $\Sigma$. Then there exists a transformation semigroup on $n+3$ states
  which is generated by $\abs{\Sigma} + 1$ elements and has at least $\abs{T}$
  different $\gJ$-classes.
\end{proposition}
\begin{proof}
  Let $T$ be a semigroup of transformations on a set of states $Q$, generated
  by a finite set $\Sigma$, and let $q_0$ be an arbitrary element from $Q$. Let
  $q_1$, $q_2$, $q_3$ be new states not in $Q$ and let $c$ be a new generator
  not in $\Sigma$. Let $U$ be the transformation semigroup on $Q \union
  \os{q_1, q_2, q_3}$ obtained by extending the transformations of $T$ as
  follows: for each $a \in \Sigma$ and $q \in Q$, let $q \cdot c = q$, $q_1
  \cdot a = q_3 \cdot a = q_3 \cdot c = q_0$, $q_1 \cdot c = q_2 \cdot a =
  q_2$, and $q_2 \cdot c = q_3$.

  Let $u, v \in \Sigma^*$ be different elements of $T$. Then $cuc$ and $cvc$
  are different in~$U$. We claim that $cuc \not\Jle cvc$ in $U$.
  For the sake of contradiction, suppose that there exist $x, y \in (\Sigma
  \union \os{c})^*$ such that $cuc = xcvcy$ in $U$. Clearly, $q_1 \cdot cuc =
  q_3 \not\in Q$.  Moreover, at least one of the words $x$ or $y$ must be
  non-empty and therefore $q_1 \cdot xcucy \in Q$. This shows that $cuc \ne
  xcvcy$, as desired.
\end{proof}
Combining the result with statements from the previous section, we obtain a
lower bound for the number of $\gJ$-classes of the transition semigroup of an
automaton.
\begin{proof}[Proof of Theorem~\ref{thm:numclasses}]
  As we mentioned before, it is well known that there exists a $3$-generator
  transformation semigroup on $n$ states of size $n^n$.
  If we first apply Proposition~\ref{prop:k-classes-lower} and then
  Proposition~\ref{prop:minimal} to $T$, we obtain the claim by
  Proposition~\ref{prop:isolated}.
  The statement extends to $\gR$-classes (resp. $\gL$-classes) because each
  $\gJ$-class contains at least one $\gR$-class (resp.\ $\gL$-class).
\end{proof}

\section{Bounds for the Length of Chains}

Let $\gK$ be any of the relations $\gR$, $\gL$ or $\gJ$.
As with the number of $\gK$-classes, the na\"{\i}ve upper bound for the length
of $\gK$-chains is given by the maximal size $n^n$ of the transformation
semigroup on $n$ states. In this section, we improve this upper bound for
$\gR$-chains and later give a lower bound that matches up to a polynomial gap.

\begin{lemma}\label{lem:inverse}
  Let $P$ be a partial transformation semigroup on a finite set $Q$ of
  cardinality $n$. Let $x, y \in P$ such that $Q \cdot x = Q \cdot xy$. Then
  $x \Req xy$.
\end{lemma}
\begin{proof}
  Let $\omega = n!$ and let $z = y^{\omega-1}$. It suffices to show that $xyz =
  x$ in $P$, \ie for all $q \in Q$, we have $q \cdot x = q \cdot xyz$.
  By assumption, the restriction of $y$ to the set $Q \cdot x$ is 
  bijective. Thus, the mapping $y^\omega$ acts as identity on $Q \cdot
  x$. 
  This yields $q \cdot xyz = q \cdot
  xy^\omega = (q \cdot x) \cdot y^\omega = q \cdot x$.
\end{proof}

\begin{proposition}\label{prop:rheight-upper}
  Let $P$ be a partial transformation semigroup on $n$ states. Then the
  $\gR$-height of $P$ is at most $2^n$.
\end{proposition}
\begin{proof}
  Let $P$ be a partial transformation semigroup on a set of states $Q$ with
  $\abs{Q} = n$.
  Let $(u_1, u_2, \dots, u_\ell)$ be an $\gR$-chain of $P$. We show that all
  sets $Q \cdot u_i$ must be pairwise distinct which yields the desired bound.
  Suppose that $Q \cdot u_i = Q \cdot u_j$ for $1 \le i < j \le \ell$. Since
  $u_j \Rl u_i$, there exists $v \in P$ with $u_i v = u_j$.
  Lemma~\ref{lem:inverse} yields $u_j \Req u_i$ which implies $u_{i+1}
  \Req u_i$, a contradiction.
\end{proof}

\subsection{Token Computations in Transformation Semigroups}

In this subsection, we introduce the building blocks for the lower bound on the
height.
A \emph{token machine} is a pair $(C, I)$ where $C$ is a finite set and $I$ is
a set of partial transformations on $C$.
The elements of the set $C$ are called \emph{cells}, subsets of $C$ are called
\emph{configurations} and the generators $I$ are called \emph{instructions}.

A \emph{program} is a finite word over the alphabet $I$ and a
\emph{computation} is a sequence
\begin{equation*}
  R_0 \xrightarrow{\iota_1} R_1 \xrightarrow{\iota_2} R_2 \cdots \xrightarrow{\iota_\ell} R_\ell
\end{equation*}
where all $R_i \subseteq C$ have the same cardinality and $R_{i-1} \cdot
\iota_i = R_i$.
The configuration $R_0$ is called \emph{initial configuration} and $R_\ell$ is
called the \emph{final configuration} of the computation.
The program $\iota_1 \iota_2 \cdots \iota_\ell$ is the \emph{label} of the
computation and $\ell$ is its \emph{length}.
It is \emph{progressing} if all configurations appearing in the computation are
pairwise distinct and for each $i \in \os{1, \dots, \ell}$ and each $\iota \in
I \setminus \os{\iota_i}$, we have $\abs{R_{i-1} \cdot \iota} < \abs{R_i}$. It
is \emph{maximal} if $\abs{R_\ell \cdot \iota} < \abs{R_\ell}$ for all $\iota
\in I$.

A language over programs $L \subseteq I^*$ is called \emph{deterministic} on a
configuration $R \subseteq C$ if $\abs{R \cdot u_1} = \abs{R} = \abs{R \cdot
u_2}$ implies $u_1 = u_2$ for all $u_1, u_2 \in L$.

The focal idea of token machines is captured in the following proposition which
states that computations in token machines naturally yield lower bounds for the
length of $\gR$-chains.

\begin{proposition}\label{prop:tokens}
  Let $(C, I)$ be a token machine and let $P$ be the partial transformation
  semigroup on $C$ generated by $I$.
  If there exists a maximal progressing computation of length $\ell$, then the
  $\gR$-height of $P$ is at least $\ell$.
\end{proposition}
\begin{proof}
  Let $R_0 \xrightarrow{\iota_1} R_1 \xrightarrow{\iota_2} R_2 \cdots
  \xrightarrow{\iota_\ell} R_\ell$ be a maximal progressing computation.
  For each $i \in \os{1, \dots, \ell}$, we let $u_i = \iota_1 \iota_2 \cdots
  \iota_i$.
  It remains to show that $(u_1, \dots, u_\ell)$ is an $\gR$-chain.
  By definition, we immediately obtain $u_{i+1} \Rle u_i$.
  Assume, for the sake of contradiction, that $u_i \Rle u_{i+1}$ for some $i
  \in \os{1, \dots, \ell-1}$, \ie there exists $v \in I^*$ with $u_i = u_{i+1}
  v$.
  Without loss of generality, we may assume that $i$ is maximal with this
  property. If $\abs{v} = 0$, then $R_i = R_0 \cdot u_i = R_0 \cdot u_{i+1} =
  R_{i+1}$, contradicting the premise of progression.
  Thus, $\abs{v} \ge 1$ and since the computation is progressing and maximal,
  we have $i < \ell-1$ and $v = \iota_{i+2} w$ for some $w \in I^*$. This
  yields $u_{i+2} w \iota_{i+1} = u_{i+1} \iota_{i+2} w \iota_{i+1} = u_{i+1} v
  \iota_{i+1} = u_i \iota_{i+1} = u_{i+1}$, contradicting the maximality of
  $i$.
\end{proof}

\subsection{Lower Bounds over a Growing Instruction Set}

Before describing the technical ingredients required in our main result, we
prove a slightly weaker statement. In contrast to the result presented later,
it relies on an alphabet that grows exponentially with the number of elements.

\begin{theorem}\label{thm:lower}
  For all even $n \in \mathbb{N}$, there exists a token machine with $n$ cells
  which admits a maximal progressing computation of length at least
  $\binom{n}{n/2} - 1$.
\end{theorem}
\begin{proof}
  Let $C = \os{1, 2, \dots, n}$. Let $\ell = \binom{n}{n/2} - 1$ and let
  $\os{R_0, R_1, \dots, R_\ell}$ be the set of $n/2$-element subsets of $C$.
  For each $i \in \os{1, \dots, \ell}$, let $\iota_i \colon R_{i-1} \to R_i$ be
  a bijection. Note that in the context of the present proof, it does not
  matter which of the $(n/2)!$ bijections is chosen; for example, one can
  always choose the unique bijection $\iota_i$ such that $\iota_i(j) <
  \iota_i(k)$ if and only if $j < k$. Each $\iota_i$ can be viewed as a partial
  transformation on $C$ which is undefined for all $c \in C \setminus R_{i-1}$.
  We now show that in the token machine $(C, I)$ with $I = \set{\iota_i}{1 \le
  i \le \ell}$, the sequence
  \begin{equation*}
    R_0 \xrightarrow{\iota_1} R_1 \xrightarrow{\iota_2} R_2 \cdots \xrightarrow{\iota_\ell} R_{\ell}
  \end{equation*}
  is a maximal progressing computation. It is a valid computation by the
  definition of the instructions $\iota_i$. Consider $i \in \os{0, \dots,
  \ell}$ and $j \in \os{1, 2, \dots, \ell} \setminus\os{i+1}$. Since $R_{j-1}
  \ne R_i$, the instruction $\iota_j$ is undefined on at least one element of
  $R_i$ and thus, $\abs{R_i \cdot \iota_j} < \abs{R_i}$. This shows that the
  computation is both progressing and maximal.
\end{proof}
The theorem has a series of interesting consequences which will be outlined in
Section~\ref{sec:main}, after proving an improved variant of the theorem with
fixed alphabet.

\subsection{Tapes and Binary Counters}

A \emph{sub-machine} of a token machine $(C, I)$ is a subset $S \subseteq C$
such that for each configuration $R$ and for each instruction $\iota \in I$ with
$\abs{R \cdot \iota} = \abs{R}$, we also have $\abs{(R \intersect S) \cdot
\iota} = \abs{R \intersect S}$. In other words, each computation stays a
computation when restricted to $S$.
The \emph{union} of two token machines $(C, I)$ and $(C', I')$ with $C \cap C'
= \emptyset$ is the token machine $(C \union C', I \union I')$ where the
instructions in $I \setminus I'$ are extended to act as identity on $C'$ and
the instructions in $I' \setminus I$ are extended to act as identity on $C$.
The cells $C$ and $C'$ of the original machines are sub-machines of the union.

An \emph{$n$-bit tape} $T$ is a token machine $(C, I)$ with $n$ cells and an
arbitrary (but fixed) order $(c_0, c_1, \dots, c_{n-1})$.
One can interpret configurations $R \subseteq C$ as bit strings $b_{n-1}
b_{n-2} \cdots b_0$ where $b_i = 1$ if and only if $c_i \in R$ and $b_i = 0$
otherwise, and think of $T$ as a ring buffer with a read/write head at position
$0$.
An instruction $\irotl^T$ can be used to move the tape head to the right (or,
actually, retain the head position but left-rotate the buffer). For each $i \in
\os{0, \dots, n-2}$, we let $c_i \cdot \irotl^T = c_{i+1}$ and $c_{n-1} \cdot
\irotl^T = c_0$. The instruction $\irotr^T$ is defined analogously and moves in
the opposite direction.
An instruction $\icz^T$ can be used to check whether the head is scanning a
zero and halt the program otherwise. It is undefined on $c_0$ and defined as
the identity on $\os{c_1, \dots, c_{n-1}}$. Conversely, the $\isync^T$
instruction is defined as the identity on $c_0$ and undefined on every other
cell.
An instruction $\imvl^T$ maps $c_0$ to $c_1$, acts as the identity on $\os{c_2,
c_3, \dots, c_{n-1}}$ and is undefined on $c_1$.  Analogously, $\imvr^T$ maps
$c_0$ to $c_{n-1}$, acts as the identity on $\os{c_1, c_2, \dots, c_{n-2}}$ and
is undefined on $c_{n-1}$.
The \emph{value} of $T$ under a configuration $R$ is $\sum_{c_i \in R} 2^i$.

An \emph{$n$-bit binary counter} $N$ is constructed as follows. Three new
$n$-bit tapes $S$, $T$ and $\ov T$ are introduced. Their cells are $(d_0, d_1,
\dots, d_{n-1})$, $(c_0, c_1, \dots, c_{n-1})$ and $(\ov c_0, \ov c_1, \dots,
\ov c_{n-1})$, respectively. Then, the union of $S$, $T$ and $\ov T$ is
constructed and the following instructions are added:

\medskip
\begin{itemize}
  \setlength\multicolsep{0pt}
  \begin{multicols}{3}
    \item $\irotl^N = \irotl^T \irotl^{\ov T} \irotl^S$,
    \item $\irotr^N = \irotr^T \irotr^{\ov T} \irotr^S$,
    \item $\icz^N = \icz^T$,
    \item $\ico^N = \icz^{\ov T}$,
    \item $\isync^N = \isync^S$,
    \item $\idisp^N = \icz^S$,
  \end{multicols}
  \item $\iinc^N$ with $\ov c_0 \cdot \iinc^N = c_0$ and $c_0 \cdot \iinc^N$
    undefined and $c \cdot \iinc^N = c$ for all $c \not\in \os{c_0, \ov c_0}$,
  \item $\idec^N$ with $c_0 \cdot \idec^N = \ov c_0$ and $\ov c_0 \cdot
    \idec^N$ undefined and $c \cdot \idec^N = c$ for all $c \not\in \os{c_0,
    \ov c_0}$.
\end{itemize}
\smallskip

\noindent
Following this, the original instructions of $S$, $T$ and $\ov T$ are removed
from $I$.
Thus, a binary counter provides exactly eight instructions.
A configuration $R$ of $N$ is \emph{valid} if $\abs{R \intersect S} = 1$ and
for each $i \in \os{0, \dots, n-1}$, we have $c_i \in R$ if and only if $\ov
c_i \not\in R$.
\begin{lemma}\label{lem:valid}
  Let $R$ be a valid configuration of a binary counter $N$ and let $u \in I^*$
  such that $\abs{R \cdot u} = \abs{R}$. Then $R \cdot u$ is a valid
  configuration of $N$.
\end{lemma}
\begin{proof}
  By induction on the length of $u$, it suffices to prove that the action of
  instructions on $R$ preserves validity.

  The instructions $\irotl^N$ and $\irotr^N$ cyclically rotate the tapes $T$,
  $\ov T$ and $S$. Thus, if $R$ is valid, then $R \cdot \irotl^N$ and $R \cdot
  \irotr^N$ are valid as well.

  For each $\iota \in \os{\icz^N, \ico^N, \isync^N, \idisp^N}$, we have either
  $\abs{R \cdot \iota} < \abs{R}$ or $R \cdot \iota = R$.

  If $\abs{R \cdot \iinc^N} = \abs{R}$, then $R$ does not contain $c_0$. If,
  moreover, $R$ is a valid configuration, it contains $\ov c_0$. But then, $R
  \cdot \iinc^N$ contains $c_0$ and does not contain $\ov c_0$. It coincides
  with $R$ on all other cells.
  Thus, $R \cdot \iinc^N$ is valid as well. By a symmetric argument, the
  instruction $\idec^N$ preserves validity.
\end{proof}
We now define three regular languages
\begin{align*}
  \Lreset^N & = \isync^N ((\icz^N \mid \idec^N) \irotr^N \idisp^N)^* (\icz^N \mid \idec^N) \irotr^N \isync^N, \\
  \Linc^N & = \isync^N (\idec^N \irotr^N \idisp^N)^* \iinc^N (\idisp^N \irotr^N)^* \isync^N \text{~and} \\
  \Ldec^N & = \isync^N (\iinc^N \irotr^N \idisp^N)^* \idec^N (\idisp^N \irotr^N)^* \isync^N.
\end{align*}

\begin{lemma}\label{lem:deterministic}
  The languages $\Lreset^N$, $\Linc^N$ and $\Ldec^N$ are deterministic on all
  valid configurations.
\end{lemma}
\begin{proof}
  Suppose there are two different words $u_1, u_2 \in \Lreset^N$ and a valid
  configuration $R$ such that $\abs{R \cdot u_1} = \abs{R}$.
  Since $\Lreset^N$ is prefix-free, there exist a unique word $p \in I^*$ and
  different instructions $\iota_1, \iota_2 \in I$ such that $u_1 \in p \iota_1
  I^*$ and $u_2 \in p \iota_2 I^*$.
  A careful analysis of the structure of the regular expression for $\Lreset^N$
  shows that either $\os{\iota_1, \iota_2} = \os{\icz^N, \idec^N}$ or
  $\os{\iota_1, \iota_2} = \os{\idisp^N, \isync^N}$.

  In the first case, we may assume without loss of generality that $\iota_1 =
  \icz^N$ and $\iota_2 = \idec^N$. From $\abs{R \cdot p \iota_1} = \abs{R \cdot
  p}$, we deduce $c_0 \not\in R \cdot p$ because $\icz^N$ is undefined on
  $c_0$.  This implies $\ov c_0 \in R \cdot p$ since $R \cdot p$ is a valid
  configuration by Lemma~\ref{lem:valid}. Since $\idec^N$ is undefined on $\ov
  c_0$, it follows that $\abs{R \cdot u_2} \le \abs{R \cdot p \iota_2} < \abs{R
  \cdot p} \le \abs{R}$.

  In the second case, we may assume that $\iota_1 = \idisp^N$ and $\iota_2 =
  \isync^N$. Since $\abs{R \cdot p \iota_1} = \abs{R \cdot p}$ and since
  $\idisp^N$ is undefined on $d_0$, we have $d_0 \not\in R \cdot p$. This
  implies $d_i \in R \cdot p$ for some $i \in \os{1, \dots, n-1}$ because $R
  \cdot p$ is valid by Lemma~\ref{lem:valid}. The instruction $\isync^N$ is
  undefined on $\os{d_1, d_2, \dots, d_{n-1}}$ which yields $\abs{R \cdot u_2}
  < \abs{R}$, as above.

  The proofs for $\Linc^N$ and $\Ldec^N$ follow by a similar reasoning.
\end{proof}
Let $R$ be a configuration of $N$.
We say that the counter is \emph{synchronized} under $R$ if $d_0 \in R$. The
\emph{value} of $N$ under $R$ is the value of $T$ under $R \intersect \os{c_0,
  c_1, \dots, c_{n-1}}$.

In addition to the eight counter instructions defined above, for any fixed
constant $k \in \os{0, \ldots, 2^n-1}$ one can define an instruction
$\ival{k}^N$ which asserts that the value of the counter equals $k$ as follows.
For each $i \in \os{0, \dots, n-1}$ with $k \bmod 2^{i+1} \ge 2^i$, we let $c_i
\cdot \ival{k}^N = c_i$ and let $\ov c_i \cdot \ival{k}^N$ be undefined.
Symmetrically, we let $\ov c_i \cdot \ival{k}^N = \ov c_i$ and $c_i \cdot
\ival{k}^N$ undefined if $k \bmod 2^{i+1} < 2^i$.

\begin{lemma}\label{lem:reset}
  Let $R$ be a valid configuration and let $u \in \Lreset^N$ such that
  $\abs{R \cdot u} = \abs{R}$.
  Then, under $R \cdot u$, the counter is synchronized and its value is zero.
\end{lemma}
\begin{proof}
  It is easy to see that each word $u \in \Lreset^N$ with $\abs{R \cdot u} =
  \abs{R}$ cyclically rotates the three tapes of $N$ exactly $n$ times and
  after each cyclic rotation, either $\icz^N$ or $\idec^N$ is applied. The
  codomains of both $\icz^N$ and $\idec^N$ do not contain $c_0$ and thus, we
  have $R \cdot u \intersect \os{c_0, c_1, \dots, c_{n-1}} = \emptyset$ which
  is equivalent to saying that the value under $R \cdot u$ is zero.
  To see that the counter is synchronized, note that applying $\isync^N$ to a
  valid configuration preserves the number of elements if and only if the
  configuration is synchronized.
\end{proof}

\begin{lemma}\label{lem:inc}
  Let $R$ be a valid configuration and let $u \in \Linc^N$ such that $\abs{R
  \cdot u} = \abs{R}$.
  If $v$ is the value of the counter under $R$ and $v'$ is its value under $R
  \cdot u$, we have $v' = v + 1 \le 2^n - 1$.
\end{lemma}
\begin{proof}
  Let us first assume that $v < 2^n - 1$. Let $i \in \os{0, \dots, n-1}$ be
  minimal such that $c_i \not\in R$ and let
  \begin{equation*}
    w = \isync^N (\idec^N \irotr^N \idisp^N)^i \iinc^N (\idisp^N \irotr^N)^{n-i} \isync^N.
  \end{equation*}

  We claim that $u = w$.
  By Lemma~\ref{lem:deterministic}, it suffices to show that $\abs{R \cdot w}
  = \abs{R}$. Let us first investigate the instructions operating on $S$. The
  word starts with an $\isync^N$ instruction, each $\idisp^N$ instruction is
  applied after $R$ has been rotated cyclically $1$ to $n-1$ times and the
  second $\isync^N$ instruction is applied after exactly $n$ cyclic rotations.
  We deduce $\abs{R \cdot \isync^N} = \abs{R}$ from $\abs{R \cdot u} =
  \abs{R}$, and thus, the counter is synchronized on both $R$ and on the
  configuration reached before the last $\isync^N$ instruction.  Moreover,
  whenever a $\idisp^N$ instruction is applied to a configuration $R'$, we have
  $d_i \in R'$ for some $i \in \os{1, \dots, n-1}$.
  Note that the case $v = 2^n - 1$ can be excluded since in order for the
  $\iinc^N$ instruction to preserve the number of elements in the
  configuration, it would have to be preceded by at least $n$ $\irotr^N
  \idisp^N$-factors and one of those factors would reduce the number of
  elements.

  The instruction $\idec^N$ is applied exactly once before each of the first
  $i$ cyclic rotations. Since $\os{c_0, c_1, \ldots, c_{i-1}} \subseteq R$, we
  have $c_0 \in R \cdot \isync^N (\idec^N \irotr^N \idisp^N)^j$ for all $j \in
  \os{0, \ldots, i-1}$.
  Moreover, since $c_i \not\in R$, we have $c_0 \not\in R \cdot \isync^N
  (\idec^N \irotr^N \idisp^N)^i$ which implies $\ov c_0 \in R \cdot \isync^N
  (\idec^N \irotr^N \idisp^N)^i$ by Lemma~\ref{lem:valid}. Consequently, the
  occurrences of $\idec^N$ and $\iinc^N$ in $w$ do not reduce the number of
  elements in the configuration.
  The above observations also show that
  \begin{equation*}
    R \cdot u = R \cdot w = \os{c_i} \union (R \intersect \os{c_{i+1}, c_{i+2}, \dots, c_{n-1}})
  \end{equation*}
  which is equivalent to the claim $v' = v + 1$.
\end{proof}
For the $\idec^N$ instruction, a symmetric version of the lemma holds.
\begin{lemma}\label{lem:dec}
  Let $R$ be a valid configuration and let $u \in \Ldec^N$ such that $\abs{R
  \cdot u} = \abs{R}$.
  If $v$ is the value of the counter under $R$ and $v'$ is its value under $R
  \cdot u$, we have $v' = v - 1 \ge 0$.
\end{lemma}

\subsection{Main Result}
\label{sec:main}

Let $n \in \mathbb{N}$ be an even number.
Let $T$ be an $n$-bit tape with cells $(t_0, t_1, \dots, t_{n-1})$. The union
of $T$ with three $\ceil{\log_2 n}$-bit counters $P$, $Q$ and $Z$ forms a token
machine, henceforth referred to as $U$.
A configuration of $U$ is \emph{valid} if it is valid when restricted to each
of the three counters.

Informally, the idea of our construction is the following: as in the proof of
Theorem~\ref{thm:lower}, we enumerate all $n/2$-element subsets of an
$n$-element set on the tape $T$.
In order to do so with a constant number of generators, this enumeration needs
to be done in a very specific way.
We say that a word $Y \in \os{0, 1}^*$ is a \emph{successor} of $X \in \os{0,
1}^*$ if there exist $p \in \os{0, 1}^*$, $i \ge 1$ and $j \ge 0$ such that $X
= p01^i 0^j$ and $Y = p10^{j+1}1^{i-1}$.
For each $m \in \os{0, 1, \dots, n}$ one can define a sequence of bit strings
$(X_0, X_1, \dots, X_\ell)$ as stated in the following lemma:

\begin{lemma}
  For all $n \in \mathbb{N}$ and $m \in \os{0, 1, \dots, n}$, there exists a
  unique sequence $(X_0, X_1, \dots, X_\ell)$ such that
  \begin{itemize}
    \item $X_0 = 0^{n-m} 1^m$,
    \item for each $k \in \os{1, \dots, \ell}$, $X_k$ is a successor of $X_{k-1}$ and
    \item $X_\ell$ does not have a successor.
  \end{itemize}
  The terms of this sequence are pairwise distinct, each term contains exactly
  $m$ occurrences of the letter $1$, and we have $\ell = \binom{n}{m}$ as well
  as $X_\ell = 1^m 0^{n-m}$.
  \label{lem:sequence}
\end{lemma}
\begin{proof}
  First observe that if a word $X \in \os{0,1}^*$ can be factorized as $X =
  p01^i 0^j$ with $p \in \os{0, 1}^*$ and $i \ge 1$ and $j \ge 0$, then this
  factorization is unique. As a consequence, the sequence defined above is
  unique and its terms are pairwise distinct. It is also easy to see that if
  $Y$ is a successor of $X$, then $X$ and $Y$ contain the same number of $1$'s.
  The remaining two properties $\ell = \binom{n}{m}$ and $X_\ell = 1^m
  0^{m-n}$ clearly hold if $n = 0$ or $m \in \os{0, n}$.

  We now assume $n \ge 1$, as well as $m \in \os{1, \dots, n-1},$ and proceed
  by induction on $n$. Let $s \in \os{0, \dots, n}$ such that $X_0, X_1, \dots,
  X_s \in 0 \os{0,1}^{n-1}$ and $X_{s+1}, X_{s+2}, \dots, X_\ell \in 1
  \os{0,1}^{n-1}$. Applying the induction hypothesis to the suffixes of length
  $n-1$ of $X_0$, $X_1$, \ldots, $X_s$, we know that $s = \binom{n-1}{m}$ and
  $X_s = 0 1^{m} 0^{(n-1)-m}$.
  This yields $X_{s+1} = 1 0^{n-m} 1^{m-1}$ and by applying induction again to
  the suffixes of $X_{s+1}$, $X_{s+2}$, \ldots, $X_\ell$, we obtain $\ell-s =
  \binom{n-1}{m-1}$ as well as $X_\ell = 1 1^{m-1} 0^{(n-1)-(m-1)} = 1^m
  0^{n-m}$.
  Note that by Pascal's rule, $\ell = \ell - s + s = \binom{n-1}{m-1} +
  \binom{n-1}{m} = \binom{n}{m}$ which concludes the proof.
\end{proof}

Note that the sequence corresponds to binary counting and deleting all counter
values not having $m$ bits equal $1$.
Since we are interested in enumerating $n/2$-element subsets, we only consider
the case $m = n/2$.
Interpreting the bit strings $X_k$ as $n/2$-element subsets of an $n$-element
set, the sequence $(X_0, X_1, \dots, X_\ell)$ describes our enumeration order.
Thus, all configurations appearing in the computation always contain $n/2$
elements when restricted to $T$. The counter $P$ keeps track of the position of
the head on $T$. It is needed for moving a block of $1$-bits as far to the
right as possible when transitioning from $X_{k-1}$ to $X_k$.
The volatile counters $Q$ and $Z$ are only used by the following macro that
checks whether the bit below the tape head of $T$ is $1$.
\begin{equation*}
  \Lco = \irotr^T ((\varepsilon \mid \icz^T \Linc^Z) \irotr^T \Linc^Q)^* \ival{n-1}^Q \ival{n/2}^Z \Lreset^Q \Lreset^Z.
\end{equation*}
Roughly speaking, a word from $\Lco$, which preserves the cardinality of the
configuration, rotates the tape $T$ cyclically $n$ times. The counter $Q$ is
used to ensure that neither more nor less rotations are performed.  After each
rotation, except for the last one, the counter $Z$ is increased
non-deterministically if the bit under the tape head is $0$. Then, the value of
$Z$ is checked to be exactly $n/2$. Since we know that the number of $0$-bits
on $T$ is $n/2$ and since the bit under the tape head cannot contribute to the
value of $Z$, this is only possible if the bit under the tape head is set.
More precisely, the following lemma holds.
\begin{lemma}\label{lem:co}
  Let $R$ be a valid configuration such that $\abs{R \intersect T} = n/2$, the
  counters $P$ and $Q$ are synchronized and the values of $P$ and $Q$ are zero.
  Then there exists a word $u \in \Lco$ with $\abs{R \cdot u} = \abs{R}$ if and
  only if $t_0 \in R$. Moreover, if such a word $u$ exists, it is unique and we
  have $R \cdot u = R$.
\end{lemma}
\begin{proof}
  For $i \in \os{0, 1, \dots, n-1}$, let $m_i = 1$ if $t_i \not\in R$ and let
  $m_i = 0$ otherwise.

  By Lemma~\ref{lem:inc}, the $\ival{n-1}^Q$ instruction in a word $w \in \Lco$
  preserves the number of elements in a valid configuration if and only if $w$
  contains exactly $n-1$ occurrences of $\Linc^Q$. Therefore, each word that
  preserves the number of elements when applied to $R$ contains the instruction
  $\irotr^N$ exactly $n$ times.
  Since each occurrence of $\Linc^Z$ is paired with a $\icz^T$ instruction,
  $\Linc^Z$ is applied at most $m_i$ times after the $i$-th rotation,
  \ie{}every program that does not reduce the number of elements when applied
  to $R$ has the form
  \begin{equation*}
    \irotr^T \prod_{i=1}^{n-1} ((\icz^T \Linc^Z)^{k_i} \irotr^T \Linc^Q) \ival{n-1}^Q \ival{n/2}^Z \Lreset^Q \Lreset^Z
  \end{equation*}
  for some $k_i \in \os{0, 1}$ with $k_i \le m_i$. Moreover, the $\ival{n/2}^Z$
  instruction preserves the cardinality of the configuration if and only if the
  sum of all $k_i$ with $1 \le i \le n-1$ equals $n/2$.
  Therefore, any choice of values $k_i$ must also satisfy
  \begin{equation*}
    n/2 = \sum_{i = 1}^{n-1} k_i \le \sum_{i = 1}^{n-1} m_i = n/2 - m_0
  \end{equation*}
  where the last equality follows from the assumption that $\abs{R \intersect
  T} = n/2$.
  This is only possible if $m_0 = 0$, \ie{}$t_0 \in R$, and $k_i = m_i$ for all
  $i \in \os{1, 2, \dots, n-1}$. By letting $k_i = m_i$ in the program above,
  we obtain the unique word $u$ such that $\abs{R \cdot u} = \abs{R}$. To see
  that $R \cdot u = R$, note that after $n$ cyclic rotations, the tape $T$
  returns to its original state. Moreover, by Lemma~\ref{lem:reset}, both $Q$
  and $Z$ are synchronized and have value zero.
\end{proof}

We also let $\Lrotl = \Ldec^P \irotl^T$ and $\Lrotr = \Linc^P \irotr^T$. The
language $L$ is now defined as $L = \Lreset^P \Lreset^Q \Lreset^Z (\Lco
\Lrotr)^* \icz^T \Lrotl (\imvl^T (L_1 \mid L_2 \mid L_3))^* \ival{n-1}^P$
with
\begin{align*}
  L_1 & = (\ival{0}^P \mid \Lrotl \icz^T \Lrotr) \Lrotr (\Lco \Lrotr)^* \icz^T \Lrotl, \\
  L_2 & = (\Lrotl \Lco)^+ \ival{0}^P (\Lco \Lrotr)^+ \icz^T \Lrotl, \\
  L_3 & = (\Lrotl \Lco)^+ \Lrotl \icz^T \Lrotr \Lrotr (K_1 \mid K_2 K_3^* K_4), \\
  K_1 & = \icz^T \Lrotl (\imvr^T \Lrotl)^* \ival{0}^P, \\
  K_2 & = \Lco \Lrotl (\imvr^T \Lrotl)^* \ival{0}^P \Lrotr (\icz^T \Lrotr)^* \Lco \Lrotr, \\
  K_3 & = \Lco \Lrotl (\imvr^T \Lrotl)^* \Lrotl \Lco \Lrotr \Lrotr (\icz^T \Lrotr)^* \Lco \Lrotr, \\
  K_4 & = \icz^T \Lrotl (\imvr^T \Lrotl)^* \Lrotl \Lco \Lrotr.
\end{align*}

\noindent
The following lemma is the technical main ingredient for
Theorem~\ref{thm:lower-const}.
\begin{lemma}\label{lem:long-word}
  There exists a valid initial configuration $R$ such that $L$ is deterministic
  on $R$. Moreover, there exists a word $u \in L$ of length at least
  $\binom{n}{n/2}$ such that $\abs{R \cdot u} = \abs{R}$.
\end{lemma}
\begin{proof}
  Let us first show that $L$ is deterministic on all configurations $R$ which
  are valid and satisfy $\abs{R \intersect T} = n/2$. Since every word in $L$
  starts with a program from $\Lreset^P \Lreset^Q \Lreset^Z$, we may also
  assume without loss of generality that each of the three counters is
  synchronized and has value zero under $R$.

  Suppose there are two different words $u_1, u_2 \in L$ and a valid
  configuration $R$ such that $\abs{R \cdot u_1} = \abs{R}$. We will show that
  then, $\abs{R \cdot u_2} < \abs{R}$.
  Since $L$ is prefix-free, there exist unique words $p, q_1, q_2 \in I^*$ and
  different instructions $\iota_1, \iota_2 \in I$ such that $u_1 = p \iota_1
  q_1$ and $u_2 = p \iota_2 q_2$.
  By Lemma~\ref{lem:deterministic} and Lemma~\ref{lem:co}, we already know that
  $\iota_1$ and $\iota_2$ do not correspond to a factor belonging to any of the
  languages $\Lreset^P$, $\Lreset^Q$, $\Lreset^Z$, $\Lrotl$, $\Lrotr$ or
  $\Lco$.
  The remaining cases are:
  \begin{enumerate}
    \item $\iota_1 = \icz^T$ and $\iota_2 q_2 \in \Lco I^*$ (or vice versa),
    \item $\iota_1 = \imvr^T$ and $\iota_2 q_2 \in \Lrotl \Lco I^*$ (or vice versa),
    \item $\iota_1 = \ival{0}^P$ and $\iota_2 q_2 \in \Lrotl I^*$ (or vice versa),
    \item $\iota_1 = \ival{0}^P$ and $\iota_2 q_2 \in \imvr^T \Lrotl I^*$ (or vice versa),
    \item $\iota_1 = \ival{n-1}^P$ and $\iota_2 q_2 \in \imvl^T (L_1 \mid L_2 \mid L_3) I^*$ (or vice versa).
  \end{enumerate}

  In the first case, since $\abs{R \cdot p \icz^T} = \abs{R \cdot p}$, we have
  $t_0 \not\in R \cdot p$.
  As $T$ is a sub-machine of $U$, we have $\abs{R \cdot p \intersect T} =
  \abs{R \intersect T} = n/2$.
  It is an invariant that $P$ and $Q$ are always synchronized and have value
  zero before and after applying a factor from $\Lco$.
  Therefore, we know by Lemma~\ref{lem:co} that $\abs{R \cdot u_2} < \abs{R
  \cdot p}$.
  In the second case, observe that $\abs{R \cdot p \imvr^T} = \abs{R \cdot p}$
  implies $t_{n-1} \not\in R \cdot p$ and after applying the prefix $r$ of
  $\iota_2 q_2$ corresponding to $\Lrotl$, we have $t_0 \not\in R \cdot pr$.
  This implies $\abs{R \cdot u_2} < \abs{R \cdot pr}$ as in the first case.

  In the third case, by $\abs{R \cdot p \ival{0}^P} = \abs{R \cdot p}$, we
  know that the value of $P$ under $R \cdot p$ is zero. Since $\iota_2 q_2 \in
  \Lrotl I^* = \Ldec^P \irotl^T I^*$, we conclude that $\abs{R \cdot u_2} <
  \abs{R \cdot p}$ by Lemma~\ref{lem:dec}.
  The fourth case is analogous to the third case and the last case is covered
  later.

  We now describe how to construct a word $u$ of the given length such that
  $\abs{R \cdot u} = \abs{R}$.
  At any time, the value of the counter $P$ describes the position of the tape
  head, \ie{}the difference between the number of right and left rotations
  performed since the beginning of the computation.
  By construction, the tape head of each tape always stays in the same place
  and the tape content is rotated or modified. However, it is often convenient
  to think of the tape head moving on a stationary tape instead. This idea is
  captured in the following definition.
  We say that a configuration $R$ \emph{encodes} a word $b_{n-1} b_{n-2} \cdots
  b_0$ with $b_i \in \os{0, 1}$ if the value of $P$ under $R$ is $v$ and for
  each $i \in \os{0, 1, \dots, n-1}$, we have $b_{i+v \bmod n} = 1$ if and only
  if $t_i \in R$.

  The initial configuration is the unique valid configuration encoding the word
  $0^{n/2} 1^{n/2}$. Then, the idea is that if some valid configuration $R'$,
  which satisfies a series of invariants described below, encodes a word $X \in
  \os{0,1}^*$, applying a word from $\imvl^T (L_1 \mid L_2 \mid L_3)$ to $R'$
  results in a configuration encoding the successor of $X$. This process can be
  repeated until we arrive at a configuration encoding $1^{n/2} 0^{n/2}$.
  Moreover, before and after applying a word from $\imvl^T (L_1 \mid L_2 \mid
  L_3)$, the tape head on $T$ always points at the leftmost bit of the
  rightmost $1$-block.
  Lemma~\ref{lem:sequence} yields the desired lower bound for the length of
  the sequence of words that corresponds to the iterated process of going from
  one encoding to its succesor.

  Let us now verify that for each configuration $R'$ corresponding to an
  encoding $X_{k-1}$ for some $k \ge 1$, there exists a word $\imvl^T w$ with
  $w \in L_1 \union L_2 \union L_3$ such that $\abs{R' \cdot \imvl^T w} =
  \abs{R'}$ and the configuration $R' \cdot \imvl^T w$ encodes $X_k$.
  By the invariant that the tape head of $T$ points to the leftmost bit of the
  rightmost $1$-block, the instruction $\imvl^T$ moves the leftmost bit of the
  rightmost $1$-block to the left, thereby replacing the encoding $p 0 1^i 0^j$
  by $p 1 0 1^{i-1} 0^j$ while the tape head stays in the same position, now
  pointing at a $0$.
  The program $w$ now needs to move the remaining $1$-block of length $i-1$ to
  the right (if applicable) and restore the invariant of the tape head pointing
  at the leftmost bit of the rightmost $1$-block.

  If $i = 1$, we apply a word from $L_1$. In that case, the new encoding
  already is $p 1 0^{j+1}$ as desired and the instructions of $w$ move the tape
  head to the left, skipping the first block of $1$-bits, moving on to the
  first $0$-bit left to the $1$-block and then returning to the leftmost
  $1$-bit of the rightmost $1$-block.
  The action is illustrated in Figure~\ref{fig:L1} for an encoding with $i = 1$
  and $j = 2$. In this case, the word preserving the cardinality of the
  configuration is
  \begin{equation*}
    w \in \Lrotl \icz^T \Lrotr \Lrotr \Lco \Lrotr \Lco \Lrotr \Lco \Lrotr \icz^T \Lrotl \subseteq L_1.
  \end{equation*}
  Note that for each configuration, only the corresponding encoding, \ie{}the
  restriction of the configuration to $T$ relative to the tape head, is
  depicted.

  \tikzstyle{tape}=[draw,minimum size=4.5mm]
  \newcommand{\extape}[3]{%
    \begin{scope}[yshift={-#1*8mm},start chain=1 going right,node distance=-0.15mm]
      \foreach \n [count=\ni] in {#2} {%
        \ifthenelse{\equal{\n}{H}}{%
          \node [on chain=1,tape,fill=gray!20,outer sep=0] (h) {$1$};
        }{%
          \ifthenelse{\equal{\n}{h}}{%
            \node [on chain=1,tape,fill=gray!20,outer sep=0] (h) {$0$};
          }{%
            \node [on chain=1,tape,outer sep=0] (\ni) {$\n$};
          }
        }
      }
      \node [on chain=1,xshift=2mm] {#3};
      \draw (1.north west) -- ++(-.2,0) decorate [decoration={zigzag, segment length=.12cm, amplitude=.02cm}] {-- ($(1.south west)+(-.2,0)$)} -- (1.south west) -- cycle;
      \node [tape,left of=1,xshift=-8mm,draw=none] {$\cdots$};
      \draw[-latex,thick] ([yshift=-2.5mm]h.south) -- ([yshift=-0mm]h.south);
    \end{scope}
  }

  \begin{figure}[t]
    \centering
    \begin{minipage}{0.49\textwidth}
      \centering
      \begin{tikzpicture}
        \extape{0}{0,1,1,0,H,0,0}{$R$}
        \extape{1}{0,1,1,1,h,0,0}{$R \cdot \imvl^T$}
        \extape{2}{0,H,1,1,0,0,0}{$R \cdot \imvl^T w$}
      \end{tikzpicture}
      \caption{Action of a word $w \in L_1$}
      \label{fig:L1}
    \end{minipage}
    \hspace{\fill}
    \begin{minipage}{0.49\textwidth}
      \begin{tikzpicture}
        \extape{0}{0,1,1,0,H,1,1}{$R$}
        \extape{1}{0,1,1,1,h,1,1}{$R \cdot \imvl^T$}
        \extape{2}{0,1,1,1,0,H,1}{$R \cdot \imvl^T w$}
      \end{tikzpicture}
      \caption{Action of a word $w \in L_2$}
      \label{fig:L2}
    \end{minipage}
  \end{figure}

  If $i > 1$ and $j = 0$, we apply a word from $L_2$. In that case, the new
  encoding already is $p 1 0 1^{i-1}$ as desired and $w$ also only moves the
  tape head back to the right position using rotation instructions similar to
  those in the case $i = 1$.
  The action is illustrated in Figure~\ref{fig:L2} for an encoding with $i = 3$
  and $j = 0$; the word is
  \begin{equation*}
    w \in \Lrotl \Lco \Lrotl \Lco \ival{0}^P \Lco \Lrotr \Lco \Lrotr \icz^T \Lrotl \subseteq L_2.
  \end{equation*}
  Again, only the encoding corresponding to each configuration is depicted.

  The remaining case is $i > 1$ and $j \ge 1$ which means that the $1$-block to
  the right of the tape head must be moved. If $j = 1$, this can be
  accomplished by applying a word from $L_3$ which ends with a program from
  $K_1$. If $j \ge 2$, one can choose a word that ends with a program from $K_2
  K_3^{j-2} K_4$.
  In the latter case, the program corresponding to $K_2$ moves the rightmost
  bit of the $1$-block to position $0$, then each program corresponding to
  $K_3$ moves one of the middle bits and the last bit is moved by a program in
  $K_4$.
  Each of the programs corresponding to $K_2$ or to $K_3$ verify that the bit
  moved to the right is not the last bit before starting the process. The
  program corresponding to $K_4$ checks that the bit moved to the right is the
  last bit by verifying that the left-hand cell contains a $0$.
  Another difference between $K_2$/$K_3$ and $K_4$ is that $K_2$/$K_3$ move the
  tape head back to the left to fetch the next bit while $K_4$ leaves the
  pointer on the last moved bit which becomes the new leftmost bit of the
  rightmost $1$-block.
  The operation of a word $w \in L_3$ is illustrated in Figure~\ref{fig:L3} for
  an encoding with $i = 5$ and $j = 2$. For better understanding, the word $w$
  is factorized as $w = w_1 w_2 w_3 w_4 w_5$ with
  \begin{align*}
    w_1 & \in (\Lrotl \Lco)^4 \Lrotl \icz^T \Lrotr \Lrotr \subseteq (\Lrotl \Lco)^+ \Lrotl \icz^T \Lrotr \Lrotr, \\
    w_2 & \in \Lco \Lrotl \imvr^T \Lrotl \imvr^T \Lrotl \ival{0}^P \Lrotr \icz^T \Lrotr \icz^T \Lrotr \Lco \Lrotr \subseteq K_2, \\
    w_3, w_4 & \in \Lco \Lrotl \imvr^T \Lrotl \imvr^T \Lrotl \Lrotl \Lco \Lrotr \Lrotr \icz^T \Lrotr \icz^T \Lrotr \Lco \Lrotr \subseteq K_3, \\
    w_5 & \in \icz^T \Lrotl \imvr^T \Lrotl \imvr^T \Lrotl \Lrotl \Lco \Lrotr \subseteq K_4
  \end{align*}
  and the intermediate results after applying each of the factors are depicted.

  \begin{figure}[t]
    \centering
    \begin{tikzpicture}
      \extape{0}{0,0,1,1,0,H,1,1,1,1,0,0}{$R$}
      \extape{1}{0,0,1,1,1,h,1,1,1,1,0,0}{$R \cdot \imvl^T$}
      \extape{2}{0,0,1,1,1,0,1,1,H,1,0,0}{$R \cdot \imvl^T w_1$}
      \extape{3}{0,0,1,1,1,0,1,H,1,0,0,1}{$R \cdot \imvl^T w_1 w_2$}
      \extape{4}{0,0,1,1,1,0,H,1,0,0,1,1}{$R \cdot \imvl^T w_1 w_2 w_3$}
      \extape{5}{0,0,1,1,1,h,1,0,0,1,1,1}{$R \cdot \imvl^T w_1 w_2 w_3 w_4$}
      \extape{6}{0,0,1,1,1,0,0,0,H,1,1,1}{$R \cdot \imvl^T w_1 w_2 w_3 w_4 w_5$}
    \end{tikzpicture}
    \caption{Action of a word $w \in L_3$}
    \label{fig:L3}
  \end{figure}
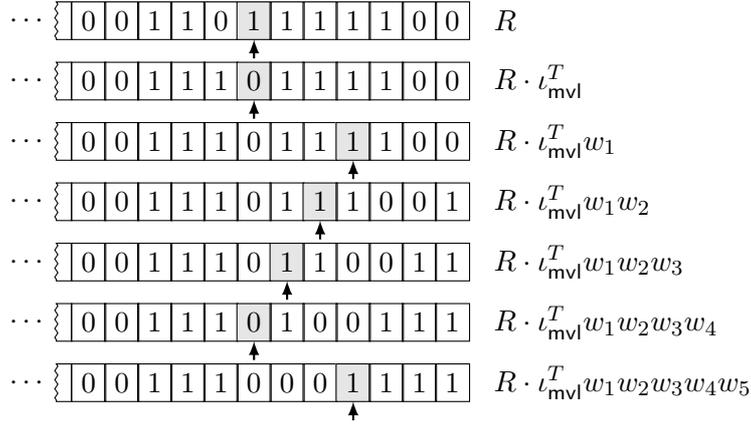

  The distinction between $K_2$ and $K_3$ is needed because we need to make
  sure that whenever moving a $1$-block, the rightmost bit eventually reaches
  position $0$ on the tape. This assertion ensures we do not return to a
  previous encoding from $1^{n/2} 0^{n/2}$: after arriving at $1^{n/2}
  0^{n/2}$, the tape head of $T$ is moved to position $n-1$. Assume we apply a
  word from $\imvl^T (L_1 \mid L_2 \mid L_3)^*$ to this configuration. Then,
  since $T$ can be thought of as a ring buffer, the encoding is replaced with
  $0 1^{n/2-1} 0^{n/2-1} 1$ and the subsequent program from $K_2$ will
  overwrite the rightmost $1$-bit, resulting in a reduction of the
  configuration size. More generally, this argument holds for any
  configuration, which is reachable from $R$ and under which the value of $P$
  is $n-1$, thereby concluding the analysis of the last case in the case
  distinction above to show that $L$ is deterministic.
\end{proof}

The last missing piece is a component that imposes the language $L$ on the
labels of valid computations.
To this end, let $\autA = (Q, I, \delta, q_0, F)$ be the minimal deterministic
automaton of $L$. We remove the sink state from $Q$ and let all transitions
leading to that state be undefined instead. Then, as long as there exists a
state which has two ingoing transitions labeled by the same letter, we create a
copy of the state and redirect one of the transitions to the copy. When
interpreting the letters of $I$ as actions on $Q$, the tuple $(Q, I)$ then
forms a token machine which we call \emph{control unit}. By construction, all
instructions are injective.

Putting the pieces together leads to the following theorem.
\begin{theorem}\label{thm:lower-const}
  For all $n \in \mathbb{N}$, there exists a token machine with $n + 9
  \ceil{\log n} + \bigO{1}$ cells and $32$ instructions which admits a maximal
  progressing computation of length at least $\binom{n}{\floor{n/2}}$.
\end{theorem}
\begin{proof}
  It suffices to prove the theorem for even numbers $n$.
  Let $V$ be the union of $U$ and the control unit.
  Any word, which is not a prefix of a word in $L$, empties the configuration
  when applied to the initial configuration $\os{q_0}$ in the control unit.
  Thus, by taking the union of the initial configuration from
  Lemma~\ref{lem:long-word} and $\os{q_0}$, we obtain a maximal progressing
  computation of the desired length in $V$.

  The only instructions required in the construction are $\irotl^T$,
  $\irotr^T$, $\imvl^T$, $\imvr^T$, $\icz^T$, $\ival{0}^P$, $\ival{n-1}^Q$,
  $\ival{n/2}^Z$ and eight additional instructions for each of the three binary
  counters.
  Since $L$ is a fixed language, the control unit has $c$ cells for a constant
  $c \in \mathbb{N}$ (independent of $n$), and $U$ has $n + 9\ceil{\log n}$
  cells: $n$ cells for the tape $T$ and $\ceil{\log n}$ cells for each of the
  three tapes of the three binary counters. Therefore, the number of cells of
  $V$ is $n + 9 \ceil{\log n} + c$.
\end{proof}

\begin{corollary}\label{cor:main}
  There exists a sequence of transformation semigroups $(T_n)_{n \in
  \mathbb{N}}$ with a fixed number of generators such that $T_n$ has $n$ states
  and the $\gR$-height (resp.\ $\gL$-height, $\gJ$-height) of $T_n$ is in
  $\Omega(2^n / n^{9.5})$.
\end{corollary}
\begin{proof}
  For the $\gR$-height, the result is an immediate consequence of
  Theorem~\ref{thm:lower-const}, Proposition~\ref{prop:tokens} and
  Proposition~\ref{prop:partial}.
  The statement also holds for $\gJ$-height because every $\gR$-chain also is a
  $\gJ$-chain; see \eg{}\cite[Proposition~1.4]{pin86}.
  An equivalent statement for the $\gL$-height follows from
  Proposition~\ref{prop:opposite} and the fact that all instructions used in
  the construction are injective.
  By Stirling's formula, we have $\binom{n}{n/2} \in \Omega(2^n / n^{0.5})$;
  see~\cite{Robbins55,Feller1968:IntroProbabilityTheory}. Thus, we obtain the
  desired bound.
  Note that the bound in Theorem~\ref{thm:lower-const} is for $n + 9
  \ceil{\log n} + \bigO{1}$ cells and not just $n$ cells.
  This
  yields the factor $n^{9}$ in the denominator.
\end{proof}

We can now prove our second main result.

\begin{proof}[Proof of Theorem~\ref{thm:main}]
  In view of Proposition~\ref{prop:minimal} and
  Proposition~\ref{prop:isolated}, the theorem immediately follows from
  Corollary~\ref{cor:main}.
\end{proof}


\newcommand{\Ju}{Ju}\newcommand{\Ph}{Ph}\newcommand{\Th}{Th}\newcommand{\Ch}{Ch}\newcommand{\Yu}{Yu}\newcommand{\Zh}{Zh}\newcommand{\St}{St}\newcommand{\curlybraces}[1]{\{#1\}}


\end{document}